\documentclass[11pt,a4paper]{article}

\usepackage[T1]{fontenc}
\usepackage[utf8]{inputenc}
\usepackage[margin=1in]{geometry}
\usepackage{mathpazo}
\usepackage[scaled=0.95]{helvet}
\usepackage{amsmath,amssymb,amsthm}
\usepackage{mathtools}
\usepackage[affil-it]{authblk}
\usepackage{microtype}
\usepackage{setspace}
\usepackage{titlesec}
\usepackage[colorlinks,linkcolor=blue,citecolor=blue,urlcolor=blue]{hyperref}

\titleformat{\section}{\normalfont\large\bfseries\sffamily}{\thesection}{0.75em}{}
\titlespacing*{\section}{0pt}{1.0\baselineskip}{0.55\baselineskip}

\makeatletter
\def\thm@space@setup{%
  \thm@preskip=0.6\baselineskip
  \thm@postskip=\thm@preskip
}
\makeatother

\newcommand{\id}{I}

\newcommand{\cart}{\mathbin{\square}}
\newcommand{\qto}{\xrightarrow{\,q\,}}

\theoremstyle{plain}
\newtheorem{theorem}{Theorem}[section]
\newtheorem{lemma}[theorem]{Lemma}

\theoremstyle{definition}
\newtheorem{definition}[theorem]{Definition}

\newtheorem{conjecture}[theorem]{Conjecture}

\title{Quantum $n$-coloring is undecidable for every $n\ge 3$}
\author{Christian Bo Kidmose-Frederiksen}
\author{Alfred Leth Nielsen}
\affil{Undergraduate students, Department of Applied Mathematics and Computer Science,\\
Technical University of Denmark}
\date{}

\begin{document}
\maketitle

\begin{abstract}
The quantum $n$-coloring problem, given graph $G$, asks whether there exists a perfect quantum strategy for the $n$-coloring game of $G$. Previously it was known that quantum $n$-coloring is undecidable only for $n=3$. We extend this to all $n\geq 3$ with one elementary
reduction. Letting $G'=(G\cart K_3)\vee K_{n-3}$, we show that $G$ is quantum $3$-colorable if and only if $G'$ is quantum $n$-colorable.
\end{abstract}

%%%%%%%%%%%%%%%%%%%%%%%%%%%%%%%%%%%%%%%%%%%%%%%%%%%%%%%%%%%%%%%%%%%%%
\section{Introduction}
\label{sec:intro}
\textit{Quantum n-colorability} was originally defined in terms of the nonlocal game in~\cite{ghfqp}. We use a more algebraic, but equivalent~\cite{CMNSW}, Definition~\ref{def:qcol}. 
For $n\in \{1,2\}$ quantum $n$-coloring coincides with the classical $n$-coloring. Thus quantum $1$-coloring is trivial and quantum $2$-coloring is decidable in
polynomial time~\cite{CMNSW}.

Harris~\cite{harris} gave a reduction from
synchronous nonlocal games to quantum $3$-coloring; combined with the known
undecidability of perfect-strategy existence arising from Slofstra's
linear-system games~\cite{slofstra}, this shows that quantum $3$-coloring is
undecidable. Prior to this work and concurrent independent work, decidability
for $n\ge4$ was open. 
This result was also independently proven by Man\v{c}inska~\cite{mancinska} and a gapped version was proven by Culf~\cite{culf}.

We write $G\qto K_n$ when Definition~\ref{def:qcol} is satisfiable for a given graph $G$. We define $G'$ given any finite graph $G$ as:
\[
G':= (G\cart K_3)\vee K_{n-3}.
\]
This mapping is polynomial time computable. Lemma~\ref{lem:reduction} shows that $G\qto K_3\iff G'\qto K_n$. This is a many-one reduction from the case $n=3$,
so quantum $n$-coloring is undecidable for every $n\ge3$.

\section{Setting and definitions}
\label{sec:prelim}
All graphs are finite, non-empty, simple and undirected; we write $u\sim v$ to denote $u$ and $v$ are adjacent. 

The \emph{Cartesian product} of graphs $G$ and $H$, denoted $G\cart H$, is defined as $V(G\cart H) = V(G)\times V(H)$ and $E(G\cart H) = \{\{(u,j),(v,i)\}| (u\sim_Gv \land j=i) \lor (u=v \land j\sim_H i)\}$. In $G\cart K_3$ each vertex $u\in V(G)$ becomes three vertices $\{(u,1),(u,2),(u,3)\}$. The set of these vertices is the \emph{fibre} of $u$, each vertex in the fibre is called a fibre vertex of $u$.

The \emph{join} of graphs $G$ and $H$, denoted $G\vee H$, is defined as the disjoint unions $V(G\vee H) = V(G)\cup V(H)$ and $E(G\vee H) = E(G) \cup E(H) \cup \{\{u,v\}|u\in V(G), v\in V(H)\}$. Any vertex $w\in V(G')=V((G\cart K_3)\vee K_{n-3})$ that is not a fibre vertex is a vertex of the clique $K_{n-3}$, enumerated by $4,\dots,n$.

We refer to the elements $V(K_n)= \{1,\dots,n\}$ as colors, when we consider the quantum $n$-colorability of some graph.

We phrase quantum coloring algebraically in terms of $\mathbb{C}^{d\times d}$, for some dimension $d\in\mathbb{N}$. A \emph{projection} is a matrix $p=p^\ast=p^2\in \mathbb{C}^{d\times d}$.

\begin{definition}[Quantum $n$-coloring]
\label{def:qcol}
Given any $d\in \mathbb{N}$, a \emph{quantum $n$-coloring} of $G$ assigns to each vertex $u\in V(G)$ a
family $\{P_{u,i}\}_{i=1}^{n}\subset \mathbb{C}^{d\times d}$ of projections such that
\[
\text{(completeness)}\quad
\sum_{i=1}^{n}P_{u,i}=\id\ \ \ \forall u\in V(G),
\qquad
\text{(orthogonality)}\quad
P_{u,i}P_{v,i}=0\ \text{ whenever }u\sim v .
\]
If there exists a quantum $n$-coloring of $G$, for some $d\in\mathbb{N}$, we say that $G$ is quantum $n$-colorable.
\end{definition}

We remark that if $\{P_{u,i}\}_{i=1}^{n}$ are the projections assigned to any vertex $u\in V(G)$ for some quantum $n$-coloring of $G$, then $P_{u,i}P_{u,j}=0$ for any $i,j\in \{1,\dots,n\}$ where $i\neq j$.

\section{The reduction}
\label{sec:proof}

\begin{lemma}
\label{lem:reduction}
For every graph $G$ and every $n\ge4$,
\[
G\qto K_3
\iff
G'\qto K_n .
\]
\end{lemma}

\begin{proof}[Proof: forward direction $(\Rightarrow)$]
In the following we are using addition of numbers modulo $3$ with representatives $\{1,2,3\}$. Let $\{P_{u,i}\}_{u\in V(G), i\in\{1,2,3\}}$ form a quantum $3$-coloring of $G$. Let
\[
E_{(u,j),i}
=
\begin{cases}
P_{u,\,i+ j}, & i\in\{1,2,3\},\\
0, & 3<i\le n,
\end{cases}
\qquad
F_{\ell,i}
=
\begin{cases}
\id, & i=\ell,\\
0, & i\ne\ell.
\end{cases}
\]
We will show that $\{E_{(u,j),i}\}_{(u,j)\in V(G\cart K_3),i\in\{1,2,\dots,n\}}\cup\{F_{\ell,i}\}_{\ell \in \{4,\dots,n\},i\in\{1,2,\dots,n\}}$ forms a quantum $n$-coloring of $G'$.

As $i\mapsto i+ j$ permutes
$\{1,2,3\}$, \textbf{completeness} is immediate:
\[
\sum_{i=1}^{n}E_{(u,j),i}
=
\sum_{i=1}^{3}P_{u,\,i+ j}
=
\sum_{k=1}^{3}P_{u,k}
=
\id,
\qquad
\sum_{i=1}^{n}F_{\ell,i}
=
\id .
\]

\paragraph{Orthogonality.} Fix a color $i$ and an edge of $G'$. An edge between fibres: $(u,j)\sim(u',j)$ for any $u\sim u'$ gives
\[
E_{(u,j),i}E_{(u',j),i}
=
P_{u,i+ j}P_{u',i+ j}
=
0.
\]
An edge contained in a fibre: $(u,j)\sim(u,j')$ for any $j\ne j'$ has
$i+ j\ne i+ j'$, thus
\[
E_{(u,j),i}E_{(u,j'),i}=P_{u,i+ j}P_{u,i+ j'}=0.
\]
For any other edge in $G'$ and any specified color at least one projection is zero, meaning they trivially are orthogonal. Hence
$\{E_{(u,j),i}\}_{(u,j)\in V(G\cart K_3),i\in\{1,2,\dots,n\}}\cup\{F_{\ell,i}\}_{\ell \in \{4,\dots,n\},i\in\{1,2,\dots,n\}}$ is a quantum $n$-coloring of $G'$, so $G'\qto K_n$.
\end{proof}

\begin{proof}[Proof: backward direction $(\Leftarrow)$]
Let $\{E_{(u,j),i}\}_{(u,j)\in V(G\cart K_3),i\in\{1,2,\dots,n\}} \cup \{F_{\ell,i}\}_{\ell \in \{4,\dots,n\},i\in\{1,2,\dots,n\}}$ be a quantum $n$-coloring of $G'$ with respect to some $\mathbb{C}^{d\times d}$. For each
color $i$ define the projection $R_i:=\sum_{\ell=4}^{n}F_{\ell,i}$. We will show that there exists a color $i^*$ for which \(I-R_{i^*}\) is nonzero. We then recover a $3$-coloring of $G$ from $\{E_{(u,j),i^*}\}_{u\in V(G),j\in \{1,2,3\}}$ on some $\mathbb{C}^{d'\times d'}$.

We show that for any chosen color, the projections of any fibre sum to a projection that is independent of the fibre.\\
Fix $v$. The vertices
\[
(v,1),(v,2),(v,3),4,\dots,n
\]
form a clique in $G'$ of size $n$, so for each color the corresponding $n$ projections
are pairwise orthogonal. Hence
\[
S_i:=\sum_{j=1}^3E_{(v,j),i}+R_i\preceq\id.
\]
Summing over the $n$ colors and using completeness at each vertex,
\[
\sum_{i=1}^{n}S_i
=
\sum_{i=1}^{n}
\Bigl(
\sum_{j=1}^{3}E_{(v,j),i}+R_i
\Bigr)
=
3\,\id+(n-3)\,\id
=
n\,\id.
\]
Thus
\[
\sum_{i=1}^n(\id - S_i)=0
\]
with every $0\preceq \id - S_i$, so each $S_i=\id$. Therefore
\begin{equation}
\label{eq:corner}
\sum_{j=1}^{3}E_{(v,j),i}
=
\id-R_i
\qquad
\text{for all }v\text{ and }i,
\end{equation}
which is independent of $v$.\\

\noindent We show that $I-R_{i^*} \neq 0$ for some color $i^*$.
\[
\sum_{i=1}^{n}R_i
=
\sum_{\ell=4}^{n}
\sum_{i=1}^{n}F_{\ell,i}
=
\sum_{\ell=4}^{n}\id
=
(n-3)I.
\]
Hence
\[
0\prec 3I=\sum_{i=1}^n\id-R_i.
\]
Therefore $0\prec\id-R_{i^\ast}$ for some color $i^\ast$.

\paragraph{Completeness}
Let $\mathcal{H'} = \operatorname{Ran}(\id-R_{i^\ast})$ where $d'=\dim(\mathcal{H'}) \leq d$.
Now defining $\{H_{u,i}\}_{u\in V(G), i\in \{1,2,3\}}$ by
\begin{align*}
    H_{u,i} = E_{(u,i),i^*}\big{|}_{\mathcal{H'}}
\end{align*}
it follows that each $H_{u,i}$ is a projection on $\mathcal{H}'$. Consider a linear isometry $\phi:\mathcal{H}'\to \mathbb{C}^{d'}$. Using $\phi$ to define $\{P_{u,i}\}_{u\in V(G), i\in \{1,2,3\}} \subset \mathbb{C}^{d'\times d'}$ by $P_{u,i}=\phi \circ H_{u,i} \circ\phi^{-1}$. It is easy to see that $P_{u,i}$ indeed are projections on $\mathbb{C}^{d'}$, and that $\sum_{i=1}^3 P_{u,i} = I_{d'}$ using~\eqref{eq:corner}.

\paragraph{Orthogonality.}

If $u\sim v$ in $G$, then $(u,i)\sim(v,i)$ in $G'$ for each $i$, so for any color $i\in\{1,2,3\}$,
\[
H_{u,i}H_{v,i} =E_{(u,i),i^*}E_{(v,i),i^*}=0.
\]
It follows $P_{u,i}P_{v,i}=0$ for any $u\sim v$ in $G$. Thus $\{P_{u,i}\}_{u\in V(G),i\in\{1,2,3\}}$ is a quantum $3$-coloring of $G$ in $\mathbb{C}^{d'\times d'}$, so $G\qto K_3$.
\end{proof}

\begin{theorem}
\label{thm:main}
For every fixed integer $n\ge3$ the corresponding quantum $n$-coloring decision problem is undecidable.
\end{theorem}

\begin{proof}
For $n=3$, this follows from Harris's reduction together with the known
undecidability of perfect-strategy existence~\cite{harris,slofstra}.
For $n\ge4$, the map
\[
G\longmapsto(G\cart K_3)\vee K_{n-3}
\]
is polynomial-time computable and, by Lemma~\ref{lem:reduction}, is a
many-one reduction from quantum $3$-coloring to quantum
$n$-coloring.
\end{proof}

\section{Concluding remarks}
\label{sec:concl}
The proof of~\autoref{thm:main} also applies to the quantum commuting-operator model, where the projections act on an arbitrary, possibly infinite-dimensional, Hilbert space $\mathcal{H}$ instead of $\mathbb{C}^{d}$.

Quantum $n$-colorability problems form a subset of quantum $H$-colorability problems. In the framework of quantum homomorphisms~\cite{qhom} quantum $H$-colorability asks if there exists a \textit{quantum homomorphism} from $G$ to $H$, and quantum $n$-colorability asks if there exists a quantum homomorphism from $G$ to $K_n$.

The complexity of quantum $H$-coloring for a general fixed graph $H$ remains
open. It is decidable for bipartite $H$, where it reduces to classical
$H$-coloring, and undecidable for the odd cycles~\cite{ciardo}; the authors expect a quantum analogue of the Hell--Ne\v{s}et\v{r}il dichotomy~\cite{hellnesetril}.

\begin{conjecture}
For every non-bipartite $H$, quantum $H$-coloring is undecidable.
\end{conjecture}

%%%%%%%%%%%%%%%%%%%%%%%%%%%%%%%%%%%%%%%%%%%%%%%%%%%%%%%%%%%%%%%%%%%%%
\section*{Acknowledgements}

We are grateful to our supervisor, David~E. Roberson, who, together
with Laura Man\v{c}inska, introduced the notion of quantum homomorphisms
in~\cite{qhom}. He introduced us to the field and guided both the mathematical
development and presentation of this work. We also learned the subject from his
PhD thesis~\cite{robersonthesis} and from his lectures.

We thank Laura Man\v{c}inska for suggesting that we study the undecidability of
quantum $n$-coloring. We are grateful to Cornelius Sevald-Krause for his insights on computational complexity and to Poul Hjorth for reading an earlier draft and offering helpful advice.

\clearpage

%%%%%%%%%%%%%%%%%%%%%%%%%%%%%%%%%%%%%%%%%%%%%%%%%%%%%%%%%%%%%%%%%%%%%

\end{document}